\documentclass[dvips,aos]{imsart}

\usepackage{amsthm,amsmath,natbib}
\usepackage{mathrsfs}
\usepackage{amssymb}
\usepackage{stmaryrd}
\usepackage{dsfont}
\usepackage{enumerate}
\usepackage[dvips]{graphicx,color}
\RequirePackage[colorlinks,citecolor=blue,urlcolor=blue]{hyperref}

\allowdisplaybreaks
\numberwithin{equation}{section}

\startlocaldefs

\theoremstyle{theorem} \newtheorem{lemma}{Lemma}
\theoremstyle{theorem} \newtheorem{theorem}{Theorem}
\theoremstyle{theorem} 
\theoremstyle{theorem} \newtheorem{proposition}{Proposition}
\theoremstyle{definition} 
\theoremstyle{remark} 
\theoremstyle{remark} 
\endlocaldefs

\begin{document}

\begin{frontmatter}

\title{Efficient Calibration for Imperfect Computer Models}
\runtitle{Efficient Calibration}


\begin{aug}
\author{\fnms{Rui} \snm{Tuo}\thanksref{t1,m1}\ead[label=e1]{tuorui@amss.ac.cn}}
\and
\author{\fnms{C. F. Jeff} \snm{Wu}\thanksref{t2,m2}\ead[label=e2]{jeffwu@isye.gatech.edu}}

\runauthor{R. Tuo and C. F. J. Wu}
\thankstext{t1}{Tuo's research is partially sponsored by the Office of Advanced Scientific Computing Research; U.S. Department of Energy, project No. ERKJ259  ``A mathematical environment for quantifying uncertainty: Integrated and optimized at the extreme scale.'' The work was performed at the Oak Ridge National Laboratory, which is managed by UT-Battelle, LLC under Contract No. De-AC05-00OR22725. Tuo's work is also supported by the National Center for Mathematics and Interdisciplinary Sciences, CAS and NSFC 11271355.}

\thankstext{t2}{Wu's research is supported by NSF DMS-1308424 and DOE DE-SC0010548.}

\affiliation{Chinese Academy of Sciences \thanksmark{m1} and Georgia Institute of Technology \thanksmark{m2}}
\address{Academy of Mathematics and Systems Science\\Chinese Academy of Sciences\\Beijing, China 100190\\ \printead{e1}}
\address{School of Industrial and Systems Engineering\\Georgia Institute of Technology\\Atlanta, Georgia 30332-0205\\ \printead{e2}}
\end{aug}


\begin{abstract}
Many computer models contain unknown parameters which need to be estimated using physical observations. \cite{tuo2014calibration} shows that the calibration method based on Gaussian process models proposed by \cite{kennedy2001bayesian} may lead to unreasonable estimate for imperfect computer models. In this work, we extend their study to calibration problems with stochastic physical data. We propose a novel method, called the $L_2$ calibration, and show its semiparametric efficiency. The conventional method of the ordinary least squares is also studied. Theoretical analysis shows that it is consistent but not efficient. Numerical examples show that the proposed method outperforms the existing ones.
\end{abstract}



\end{frontmatter}


\section{Introduction}

Computer simulations are widely used by researchers and engineers to understand, predict, or control complex systems. Many physical phenomena  and processes can be modeled with mathematical tools, like partial differential equations. These mathematical models are solved by numerical algorithms like the finite element method. For example, computer simulations can help predict the trajectory of a storm. In engineering, computer simulations have become more popular and sometimes indispensable in product and process designs. The design and analysis of experiments is a classic area of statistics. A new field has emerged, which considers the design and analysis for experiments in computer simulations, commonly referred to as ``computer experiments''. Unlike the physical experiments, computer experiments are usually deterministic. In addition, the input variables for computer experiments usually take real values, not discrete levels as in many physical experiments. Therefore, interpolation methods are widely used in computer experiments, while conventional methods, like ANOVA or regression models, are used much less often.

In many computer experiments, some of the input parameters represent certain inherent attributes of the physical system. The true values of these variables are unknown because there may not be enough knowledge about the physical systems. For instance, in underground water simulations, the soil permeability is an important input parameter, but its true value is usually unknown. A standard approach to identify the unknown model parameters is known as \textit{calibration}. To calibrate the unknown parameters, one need to run the computer model under different model parameters, and run some physical experiments. The basic idea of calibration is to find the combination of the model parameters, under which the computer outputs match the physical responses.

One important topic in the calibration of computer models is to tackle the model uncertainty. Most physical models are built under certain assumptions or simplifications, which may not hold in reality. As a result, the computer output can rarely fit the physical response perfectly, even if the true values of the calibration parameters are known. We call such computer models \textit{imperfect}. \cite{kennedy2001bayesian} first discusses this model uncertainty problem and proposes a Bayesian method, which models the discrepancy between the physical process and the computer output as a Gaussian process. Because of the importance of calibration for computer models, the Kennedy-O'Hagan's approach has been widely used, including hydrology, radiological protection, cylinder implosion, spot welding, micro-cutting, climate prediction and cosmology. See \cite{higdon2004combining,higdon2008computer,higdon2013computer}, \cite{bayarri2007computer,bayarri2007framework}, \cite{joseph2009statistical}, \cite{wang2009bayesian}, \cite{han2009simultaneous}, \cite{goldstein2004probabilistic} and \cite{murphy2007methodology}, \cite{goh2013prediction}.

\cite{tuo2014calibration} studies the asymptotic properties of the calibration parameter estimators given by \cite{kennedy2001bayesian} and shows that their method can lead to unreasonable estimate.
The first theoretical framework for calibration problems of the Kennedy-O'Hagan type is established by \cite{tuo2014calibration} under the assumption that the physical responses have no random error. This assumption is needed to make the mathematical analysis for a version of the Kennedy-O'Hagan's approach feasible. Given the fact that the responses in physical experiments are rarely deterministic, it is necessary to extend the study to cases where the physical responses have measurement or observational errors. For convenience, we use the term ``stochastic physical experiments'' to denote physical responses with random errors.

In \cite{tuo2014calibration}, the theory of native spaces is used to derive the convergence rate for calibration with deterministic physical systems. Because of the random error in the current context, the interpolation theory fails to work. In this work we will mainly use mathematical tools of weak convergence, including the limiting theory of empirical processes.

The main theme of this article is to propose a general framework for calibration and provide an efficient estimator for the calibration parameter.
We utilize a nonparametric regression method to model the physical outputs. Similar models are also considered in the literature of response surface methodology. See \cite{myers1999response} and \cite{anderson2005some}. To estimate the calibration parameter, we extend the $L_2$ calibration method proposed by \cite{tuo2014calibration} to the present context. This novel method is proven to be semiparametric efficient when the measurement error follows a normal distribution.
A conventional method, namely, the ordinary least squares method, is also studied,
and shown to be consistent but not efficient.

This paper is organized as follows. In Section 2, we extend the $L_2$ projection defined by \cite{tuo2014calibration} to stochastic systems and propose the $L_2$ calibration method in the current context. In Section 3, the asymptotic behavior for $L_2$ calibration is studied. In Section 4, we consider the ordinary least squares method. The proposed method is illustrated and its performance studied in two numerical examples in Section 5. Concluding remarks are given in Section 6.

\section{$L_2$ Projection for Systems with Stochastic Physical Experiments}


Let $\Omega$ denote the region of interest for the control variables, which is a convex and compact subset of $\mathbf{R}^d$. Let $x_1,\ldots,x_n$ be a set of points on $\Omega$. Suppose the physical experiment is conducted once on each $x_i$, with the corresponding response denoted by $y^p_i$, for $i=1,\ldots,n$, where the superscript $p$ stands for ``physical''. In this work, we assume the physical system is stochastic, that is, the physical responses have random measurement or observational errors. To incorporate this randomness, we consider the following nonparametric model:
\begin{eqnarray}
y_i=\zeta(x_i)+e_i,\label{lpls}
\end{eqnarray}
where $\zeta(\cdot)$ is an unknown deterministic function and $\{e_i\}_{i=1}^n$ is a sequence of independent and identically distributed random variables with $E e_i=0$ and $E e_i^2=\sigma^2<+\infty$. This model is also adopted by \cite{kennedy2001bayesian}, where $\zeta(\cdot)$ is called the \textit{true process}. In addition, \citeauthor{kennedy2001bayesian} assumes that $e_i$'s follow a normal distribution. Such a distribution assumption will be slightly relaxed in our theoretical analysis.

Let $\Theta$ be the parameter space for the calibration parameter $\theta$. Suppose $\Theta$ is a compact region in $\mathbf{R}^q$. Denote the output of the deterministic computer code at $(x,\theta)\in\Omega\times\Theta$ by $y^s(x,\theta)$, where the superscript $s$ stands for ``simulation''.
In the frequentist framework of calibration established by \cite{tuo2014calibration}, the concept of \textit{$L_2$ projection} plays a central role. Because the ``true'' calibration parameter (as stated in \cite{kennedy2001bayesian}) is unidentifiable, \cite{tuo2014calibration} defines the purpose of calibration as that of finding the $L_2$ projection $\theta^*$ which minimizes the $L_2$ distance between the physical response surface and the computer outputs as a function of the control variables. In the present context, the physical responses are observed with errors. A good definition of the ``true'' value of $\theta$ should exclude the uncertainty in $y^p$. Thus we suggest the following definition for the $L_2$ projection using the true process $\zeta(\cdot)$:
\begin{eqnarray}
\theta^*:=\operatorname*{argmin}\limits_{\theta\in\Theta}\|\zeta(\cdot)-y^s(\cdot,\theta)\|_{L_2(\Omega)}.\label{l2projection}
\end{eqnarray}
The main focus of this article is on the statistical inference for $\theta^*$.

\subsection{$L_2$ Calibration}\label{Sec L2stochastic}
In this section, we will extend the $L_2$ calibration method proposed by \cite{tuo2014calibration} to the present context. Since \citeauthor{tuo2014calibration} assumes that the physical experiment is deterministic, they use the kernel interpolation method to approximate the physical response surface.
Because of the existence of the random error in (\ref{lpls}), the kernel interpolation can perform poorly because interpolation methods generally suffer from the problem of overfitting.

In spatial statistics, the effect of the random error is usually modeled with a white noise process, which is also referred to as a nugget term in the kriging modeling \citep{cressie1993statistics}.
In the computer experiment literature, it is also common to use the nugget term in Gaussian process modeling to tackle the numerical instability problems \citep{gramacy2010cases,peng2013choice}.

Let $z(\cdot)$ be a Gaussian process with mean zero and covariance function $\Phi(\cdot,\cdot)$. Suppose $\{(x_i,y_i)\}_{i=1}^y$ are obtained, which satisfy $y_i=z(x_i)+\epsilon_i$ with $\epsilon_i$'s being i.i.d. and distributed as $N(0,\sigma^2)$. Then the predictive mean of $z(\cdot)$ is given by
\begin{eqnarray}
\hat{z}(x)=\sum_{i=1}^n u_i\Phi(x_i,x),\label{representer2}
\end{eqnarray}
where $u=(u_1,\ldots,u_n)^\text{T}$ is the solution to the linear system
\begin{eqnarray}
Y=(\mathbf{\Phi}+\sigma^2 \mathbf{I})u,\label{representer1}
\end{eqnarray}
with $Y=(y_1,\ldots,y_n)^\text{T}$ and $\mathbf{\Phi}=(\Phi(x_i,x_j))_{i j}$.
By the representer Theorem \citep{wahba1990spline,scholkopf2001generalized}, $\hat{z}(x)$ given by (\ref{representer2}) and (\ref{representer1}) is the solution of the following minimization problem with some $\lambda>0$:
\begin{eqnarray}
\operatorname*{argmin}_{f\in\mathcal{N}_\Phi(\Omega)}\frac{1}{n}\sum_{i=1}^n(y_i-f(x_i))^2+\lambda\|f\|^2_{\mathcal{N}_\Phi(\Omega)},\label{RKHSreg}
\end{eqnarray}
where $\|\cdot\|_{\mathcal{N}_\Phi(\Omega)}$ is the norm of the \textit{reproducing kernel Hilbert space} $\mathcal{N}_\Phi(\Omega)$ generated by the kernel function $\Phi$. We refer to \cite{wendland2005scattered} and \cite{wahba1990spline} for detailed discussions about these spaces. The solution to (\ref{RKHSreg}) is referred to as the \textit{nonparametric regressor} in the reproducing kernel Hilbert space \citep{berlinet2004reproducing}.

Now we are ready to define the $L_2$ calibration method for systems with stochastic physical experiments. Suppose the physical experiment is conducted over a design set $\{x_1,\ldots,x_n\}$. Define
\begin{eqnarray}
\hat{\zeta}:=\operatorname*{argmin}_{f\in\mathcal{N}_\Phi(\Omega)}\frac{1}{n}\sum_{i=1}^n(y^p_i-f(x_i))^2+ \lambda\|f\|^2_{\mathcal{N}_\Phi(\Omega)},\label{smoothing}
\end{eqnarray}
where the smoothing parameter $\lambda$ can be chosen using certain model selection criterion, e.g., generalized cross validation (GCV). See \cite{wahba1990spline}. 
We define the $L_2$ calibration for $\theta$ as
\begin{eqnarray*}
\hat{\theta}^{L_2}:=\operatorname*{argmin}_{\theta\in\Theta}\|\hat{\zeta}(\cdot)-\hat{y}^s(\cdot,\theta)\|_{L_2(\Omega)},
\end{eqnarray*}
where $\hat{y}^s(\cdot,\cdot)$ is an \textit{emulator} for the computer code $y^s(\cdot,\cdot)$. In this work, the emulator for the computer model can be constructed by any method provided that it approximates $y^s$ well. For instance, $\hat{y}^s$ can be constructed by the radial basis function approximation \citep{wendland2005scattered}, Gaussian process models \citep{santner2003design} or the polynomial chaos approximation \citep{xiu2010numerical}.

\section{Asymptotic Results for $L_2$ Calibration}

We now consider the asymptotic behavior of $\hat{\theta}^{L_2}$ as the sample size $n$ becomes large. For mathematical rigor, we write
\begin{eqnarray}
\hat{\zeta}_n=\operatorname*{argmin}_{f\in\mathcal{N}_\Phi(\Omega)}\frac{1}{n}\sum_{i=1}^n(y^p_i-f(x_i))^2+ \lambda_n \|f\|^2_{\mathcal{N}_\Phi(\Omega)},\label{smoothingn}
\end{eqnarray}
for all sufficiently large $n$, where $\{\lambda_n\}_{i=1}^n$ is a prespecified sequence of positive values. For the ease of mathematical treatment, we assume the $x_i$'s are a sequence of random samples rather than fixed design points. We also write the $L_2$ calibration estimator indexed by $n$ as
\begin{eqnarray}
\hat{\theta}^{L_2}_n:=\operatorname*{argmin}_{\theta\in\Theta}\|\hat{\zeta}_n(\cdot)-\hat{y}^s_n(\cdot,\theta)\|_{L_2(\Omega)},\label{L2stochastic}
\end{eqnarray}
where the emulator $\hat{y}^s$ is also indexed by $n$. We assume that $\hat{y}^s_n$ has increasing approximation power as $n$ becomes large.


\subsection{Asymptotic Results for $\hat{\zeta}_n$}

Before stating the asymptotic results for $\hat{\theta}_n^{L_2}$, we need first to show that $\hat{\zeta}_n$ tend to $\zeta$. To study the convergence, we need some additional definitions from the theory of empirical processes \citep{kosorok2008introduction}. For function space $\mathcal{F}$ over $\Omega$, define the covering number $N(\delta,\mathcal{F},\|\cdot\|_{L_\infty(\Omega)})$ as the smallest value of $N$ for which there exist functions $f_1,\ldots,f_N$, such that for each $f\in\mathcal{F}$, $\|f-f_j\|_{L_\infty(\Omega)}\leq\delta$ for some $j\in\{1,\ldots,N\}$. The $L_2$ covering number with bracketing $N_{[\,]}(\delta,\mathcal{F},\|\cdot\|_{L_2(\Omega)})$ is the smallest value of $N$ for which there exist $L_2$ functions $\{f_1^L,f_1^U,\ldots,f_N^L,f_N^U\}$ with $\|f_j^U-f_j^L\|_{L_2(\Omega)}\leq\delta, j=1,\ldots,N$ such that for each $f\in\mathcal{F}$ there exists a $j$ such that $f_j^L\leq f\leq f_j^U$. 

We now state a result for general nonparametric regression.
Suppose $\mathcal{F}$ is a space of functions over a compact region $\Omega$ equipped with a norm $\|\cdot\|$. Suppose the true model is
\begin{eqnarray}
y_i=f_0(x_i)+\epsilon_i,\label{npmodel}
\end{eqnarray}
and $x_i$ are i.i.d. from the uniform distribution $U(\Omega)$ over $\Omega$. In addition, the sequences $\{x_i\}$ and $\{e_i\}$ are independent and $e_i$ has zero mean. We use ``$\preceq$'' to denote that the left side is dominated by the right side up to a constant. Let
\begin{eqnarray}
\hat{f}_n=\operatorname*{argmin}_{f\in\mathcal{F}}\frac{1}{n}\sum_{i=1}^n(y_i-f(x_i))^2+\lambda_n\|f\|^2,\label{npestimator}
\end{eqnarray}
for some $\lambda_n>0$.

\begin{lemma}\label{Le smoothingspline}
Under the model (\ref{npmodel}), suppose $f_0\in\mathcal{F}$. Let $\mathcal{F}(\rho):=\{f\in\mathcal{F}:\|f\|\leq \rho\}$. Suppose there exists $C_0>0$ such that $E[\exp(C_0 |e_i|)]<\infty$.
Moreover, there exists $0<\tau<2$ such that
\begin{eqnarray*}
\log N_{[\,]}(\delta,\mathcal{F}(\rho),\|\cdot\|_{L_2(\Omega)})\preceq \rho^\tau \delta^{-\tau},
\end{eqnarray*}
for all $\delta,\rho>0$.
Then if $\lambda^{-1}_n=O(n^{2/(2+\tau)})$, the estimator $\hat{f}_n$ given by (\ref{npestimator}) satisfies
\begin{eqnarray*}
\|\hat{f}_n\|=O_p(1),\text{ and } \|\hat{f}_n-f_0\|_{L_2(\Omega)}=O_p(\lambda^{1/2}_n).
\end{eqnarray*}

\end{lemma}

\begin{proof}
See \cite{van2000empirical}.
\end{proof}

The covering numbers for some reproducing kernel Hilbert spaces have been calculated accurately in the literature.
For instance, consider a Mat\'{e}rn kernel function given by
\begin{eqnarray}\label{matern}
\Phi(s,t;\nu,\phi)=\frac{1}{\Gamma(v)2^{\nu-1}}\left(2\sqrt{\nu}\phi \|s-t\|\right)^\nu K_\nu\left(2\sqrt{\nu}\phi\|s-t\|\right),
\end{eqnarray}
with $\nu\geq 1$ \citep{stein1999interpolation,santner2003design}. The reproducing kernel Hilbert space generated by this kernel function is equal to the (fractional) Sobolev space $H^{\nu+d/2}(\Omega)$, and $\|\cdot\|_{\mathcal{N}_\Phi(\Omega)}$ and $\|\cdot\|_{H^{\nu+d/2}(\Omega)}$ are equivalent. See Corollary 1 of \cite{tuo2014calibration}. Let $H^\mu(\Omega,\rho):=\{f:\|f\|_{H^\mu(\Omega)}\leq \rho\}$. \cite{edmunds2008function} proves that for $\mu>d/2$, the covering number of $H^\mu(\Omega,\rho)$ is bounded by
\begin{eqnarray*}
\log N(\delta,H^\mu(\Omega,\rho),\|\cdot\|_{L_\infty(\Omega)})\leq \Big(\frac{C \rho}{\delta}\Big)^{d/\mu},
\end{eqnarray*}
where $C$ is independent of $\rho$ and $\delta$.
To calculate the $L_2$ metric entropy with bracketing, we note the fact that every $f, f'\in H^\mu(\Omega,\rho)$ with $\|f-f'\|\leq \delta$ satisfy the inequality $f'-\delta\leq f\leq f'+\delta$. Thus the union of the $\delta$-balls centered at $f_1,\ldots,f_n$ is covered by the union of the ``brackets'' $[f_1-\delta,f_1+\delta],\ldots,[f_n-\delta,f_n+\delta]$, which, together with the definition of the covering number and the $L_2$ covering number with bracketing, implies that
\begin{eqnarray}
\log N_{[\,]}(2\delta \sqrt{Vol(\Omega)},H^\mu(\Omega,\rho),\|\cdot\|_{L_2(\Omega)})\leq \Big(\frac{C \rho}{\delta}\Big)^{d/\mu},\label{soboleventropy}
\end{eqnarray}
where $Vol(\Omega)$ denotes the volume of $\Omega$, and $2\delta \sqrt{Vol(\Omega)}$ is the $L_2(\Omega)$ norm of the function $2\delta$. Then by applying Lemma \ref{Le smoothingspline}, the following result can be obtained after direct calculations.

\begin{proposition}\label{Pr nppart}
Under the model (\ref{npmodel}), suppose $f_0\in\mathcal{F}=\mathcal{N}_\Phi(\Omega)$ and $\mathcal{N}_\Phi(\Omega)$ can be embedded into $H^\mu(\Omega)$ with $\mu>d/2$. Choose $\|\cdot\|=\|\cdot\|_{\mathcal{N}_\Phi(\Omega)}$. Then for $\lambda^{-1}_n=O(n^{2\mu/(2\mu+d)})$, the estimator $\hat{f}_n$ given by (\ref{npestimator}) satisfies
\begin{eqnarray*}
\|\hat{f}_n\|_{\mathcal{N}_\Phi(\Omega)}=O_p(1), \text{ and } \|\hat{f}_n-f_0\|_{L_2(\Omega)}=O_p(\lambda^{1/2}_n).
\end{eqnarray*}
\end{proposition}

In Proposition \ref{Pr nppart}, one can choose $\lambda\asymp n^{-2\mu/(2\mu+d)}$ to obtain the best convergence rate $\|\hat{f}_n-f_0\|_{L_2(\Omega)}=O_p(n^{-\mu/(2\mu+d)})$, where ``$\asymp$'' denotes that its left and the right sides have the same order of magnitude. This rate is known to be optimal \citep{stone1982optimal}.

\subsection{Asymptotic Normality}

The main purpose of calibration is to estimate the calibration parameter $\theta^*$. In this section, we will prove some convergence properties of the $L_2$ calibration: its convergence rate is given by $\|\hat{\theta}_{n}^{L_2}-\theta^*\|=O_p(n^{-1/2})$ and the distribution of $\sqrt{n}(\hat{\theta}_{n}^{L_2}-\theta^*)$ tends to normal as $n\rightarrow \infty$ under certain conditions. This is a nontrivial result because the convergence rate for the nonparametric part $\|\hat{\zeta}_n(\cdot)-\zeta\|_{L_2(\Omega)}$ is generally slower than $O_p(n^{-1/2})$ (see Proposition \ref{Pr nppart}).

We first list necessary conditions for the convergence result, which are grouped in three categories.

The first group consists of regularity conditions on the model. For any $\theta\in\Theta\subset\mathbf{R}^q$, write $\theta=(\theta_1,\ldots,\theta_q)$.

\begin{enumerate}[{A}1{:}]
\item The sequences $\{x_i\}$ and $\{e_i\}$ are independent; $x_i$'s are i.i.d. from $U(\Omega)$; and $\{e_i\}$ is a sequence of i.i.d. random variables with zero mean and finite variance.
\item $\theta^*$ is the unique solution to (\ref{l2projection}), and is an interior point of $\Theta$.
\item $\sup_{\theta\in\Theta}\|y^s(\cdot,\theta)\|_{L_2(\Omega)}< +\infty$.
\item $V:=E\left[\frac{\partial^2}{\partial\theta\partial\theta^\text{T}}(\zeta(x_1)-y^s(x_1,\theta^*))^2\right]$ is invertible.
\item There exists a neighborhood $U\subset\Theta$ of $\theta^*$, such that
\begin{eqnarray*}
\sup_{\theta\in U}\left\|\frac{\partial y^s}{\partial \theta_j}(\cdot,\theta)\right\|_{\mathcal{N}_\Phi(\Omega)}<+\infty,\frac{\partial^2 y^s}{\partial \theta_j\partial \theta_k}(\cdot,\cdot)\in C(\Omega\times U),
\end{eqnarray*}
for all $\theta\in U$ and $j,k=1,\ldots,q$.
\end{enumerate}


Next we need some conditions on the nonparametric part.

\begin{enumerate}[{B}1{:}]
\item $\zeta\in\mathcal{N}_\Phi(\Omega)$ and $\mathcal{N}_\Phi(\Omega,\rho)$ is Donsker for all $\rho>0$.
\item $\|\hat{\zeta}-\zeta\|_{L_2(\Omega)}=o_p(1)$.
\item $\|\hat{\zeta}\|_{\mathcal{N}_\Phi(\Omega)}=O_p(1)$.
\item $\lambda_n=o_p(n^{-1/2})$.
\end{enumerate}

The Donsker property is an important concept in the theory of empirical processes. For its definition and detailed discussion, we refer to \cite{van1996weak} and \cite{kosorok2008introduction}. One major result is that a class of functions over domain $\Omega$, denoted as $\mathcal{F}$, is Donsker, if
\begin{eqnarray*}
\int_0^\infty\sqrt{\log N_{[\,]}(\delta,\mathcal{F},L_2(\Omega))}d \delta<+\infty.
\end{eqnarray*}
Thus from (\ref{soboleventropy}) we can see that if $\mathcal{N}_\Phi(\Omega)$ can be embedded into $H^\mu(\Omega)$ for some $\mu>d/2$, $\mathcal{N}_\Phi(\Omega)$ is Donsker. Actually if we further assume condition A1 and $E[\exp(C|e_i|)]<+\infty$ for some $C>0$, the conditions of Proposition \ref{Pr nppart} are satisfied. Then by choosing a suitable sequence of $\{\lambda_n\}$, say $\lambda\asymp n^{-2\mu/(2\mu+d)}$, one can show that condition B4 holds and condition B2 and B3 are ensured by Proposition \ref{Pr nppart}.

Finally we need to assume some convergence properties for the emulator. In this work, we assume that the approximation error caused by emulating the computer experiment is negligible compared to the estimation error caused by the measurement error in the physical experiment. Under this assumption, the asymptotic behavior of $\hat{\theta}^{L_2}_n-\theta^*$ is determined by the central limit theorem. Given that computer experiment is usually much cheaper to run than physical experiment, such an assumption is reasonable because the size of computer runs is in general much larger than the size of physical trials.

\begin{enumerate}[{C}1{:}]
\item $\|\hat{y}^s_n-y^s\|_{L_\infty(\Omega\times\Theta)}=o_p(n^{-1/2})$.
\item $\left\|\frac{\partial\hat{y}^s}{\partial\theta_i}-\frac{\partial y^s}{\partial \theta_i}\right\|_{L_\infty(\Omega\times\Theta)}=o_p(n^{-1/2})$, for $i=1,\ldots,q$.
\end{enumerate}


Now we are ready to state the main theorem of this section on the asymptotic normality of the $L_2$ calibration.

\begin{theorem}\label{Th AN}
Under conditions A1-A5, B1-B4, and C1-C2, we have
\begin{eqnarray}
\hat{\theta}_n^{L_2}-\theta^*=-2 V^{-1}\left\{ \frac{1}{n}\sum_{i=1}^n e_i \frac{\partial y^s}{\partial \theta}(x_i,\theta^*)\right\}+o_p(n^{-1/2}),\label{normality}
\end{eqnarray}
where $V$ is defined in condition A4.
\end{theorem}

\begin{proof}
We first prove that $\hat{\theta}_n\operatorname*{\rightarrow}\limits^p \theta^*$. From the definitions of $\theta^*$ and $\hat{\theta}_n^{L_2}$ in (\ref{l2projection}) and (\ref{L2stochastic}), it suffices to prove that $\|\hat{\zeta}_n(\cdot)-\hat{y}^s_n(\cdot,\theta)\|_{L_2(\Omega)}$ converges to $\|\zeta(\cdot)-y^s(\cdot,\theta)\|_{L_2(\Omega)}$ uniformly with respect to $\theta\in\Theta$ in probability, which is ensured by
\begin{eqnarray}
&&\|\hat{\zeta}_n(\cdot)-\hat{y}^s(\cdot,\theta)\|_{L_2(\Omega)}^2-\|\zeta(\cdot)-y^s(\cdot,\theta)\|_{L_2(\Omega)}^2\nonumber\\
&=&\int_{\Omega}\left(\hat{\zeta}_n(z)-\zeta(z)-\hat{y}^s(z,\theta)+y^s(z,\theta)\right)\left(\hat{\zeta}_n(z)+\zeta(z)- y^s(z,\theta)-\hat{y}^s(z,\theta)\right) d z\nonumber\\
&\leq&\left(\|\hat{\zeta}_n-\zeta\|_{L_2(\Omega)}+\|\hat{y}^s(\cdot,\theta)-y^s(\cdot,\theta)\|_{L_2(\Omega)}\right)\cdot\nonumber\\ &&\left(\|\hat{\zeta}_n(\cdot)\|_{L_2(\Omega)}+\|\zeta(\cdot)\|_{L_2(\Omega)}+ \|y^s(\cdot,\theta)\|_{L_2(\Omega)}+\|\hat{y}^s(\cdot,\theta)\|_{L_2(\Omega)}\right), \label{thetan}
\end{eqnarray}
where the inequality follows from the Schwarz inequality and the triangle inequality.
Denote the volume of $\Omega$ by $Vol(\Omega)$. It is easily seen that
\begin{eqnarray*}
\|f\|_{L_2(\Omega)}\leq Vol(\Omega)\|f\|_{L_\infty(\Omega)}
\end{eqnarray*}
holds for all $f\in L_\infty(\Omega)$. Thus
\begin{eqnarray}
\|\hat{y}^s(\cdot,\theta)-y^s(\cdot,\theta)\|_{L_2(\Omega)}&\leq& \sqrt{Vol(\Omega)}\|\hat{y}^s(\cdot,\theta)-y^s(\cdot,\theta)\|_{L_\infty(\Omega)}\nonumber\\ &\leq& \sqrt{Vol(\Omega)}\|\hat{y}^s-y^s\|_{L_\infty(\Omega\times\Theta)}.\label{uniform}
\end{eqnarray}
Additionally, we have
\begin{eqnarray}
&&\|\hat{\zeta}_n\|_{L_2(\Omega)}\leq Vol(\Omega)\|\hat{\zeta}_n\|_{L_\infty(\Omega)}
=Vol(\Omega)\sup_{x\in\Omega}\left\langle\hat{\zeta}_n,\Phi(\cdot,x)\right\rangle_{\mathcal{N}_\Phi(\Omega)}\nonumber\\
&\leq& Vol(\Omega)\|\hat{\zeta}_n\|_{\mathcal{N}_\Phi(\Omega)}\sup_{x\in\Omega}\|\Phi(\cdot,x)\|_{\mathcal{N}_\Phi(\Omega)}
=Vol(\Omega)\|\hat{\zeta}_n\|_{\mathcal{N}_\Phi(\Omega)}.\label{l2native}
\end{eqnarray}
Combining (\ref{uniform}), (\ref{l2native}), B2, and C1, we have that (\ref{thetan}) convergence to 0 uniformly with respect to $\theta\in\Theta$, which yields the consistency of $\hat{\theta}^{L_2}_n$.

Since $\hat{\theta}$ minimizes (\ref{L2stochastic}), following A1, A2 and A5 we have
\begin{eqnarray*}
0&=&\frac{\partial}{\partial\theta}\|\hat{\zeta}_n(\cdot)-\hat{y}^s(\cdot,\hat{\theta}_n^{L_2})\|^2_{L_2(\Omega)}\nonumber\\
&=&2\int_\Omega\left(\hat{\zeta}_n(z)-\hat{y}^s(z,\hat{\theta}_n^{L_2})\right)\frac{\partial \hat{y}^s}{\partial \theta}(z,\hat{\theta}_n^{L_2})d z,
\end{eqnarray*}
which, together with B2, C1 and C2, implies
\begin{eqnarray}
\int_\Omega\left(\hat{\zeta}_n(z)-y^s(z,\hat{\theta}_n^{L_2})\right)\frac{\partial y^s}{\partial \theta}(z,\hat{\theta}_n^{L_2})d z=o_p(n^{-1/2}).\label{eq theta}
\end{eqnarray}
Let $L_n(f)=n^{-1}\sum_{i=1}^n(y^p_i-f(x_i))^2+\lambda_n\|f\|^2_{\mathcal{N}_\Phi(\Omega)}$. By (\ref{smoothing}), $\hat{\zeta}_n$ minimizes $L_n$ over $\mathcal{N}_\Phi(\Omega)$. Since $\hat{\theta}^{L_2}_n$ is consistent, by A5, $\frac{\partial y^s}{\partial \theta_j}(\cdot,\hat{\theta}^{L_2}_n)\in\mathcal{N}_\Phi(\Omega)$ for $j=1,\ldots,q$ and sufficiently large $n$. Thus we have
\begin{eqnarray}
0&=&\frac{\partial}{\partial t} L\Big(\hat{\zeta}_n(\cdot)+t\frac{\partial y^s}{\partial \theta_j}(\cdot,\hat{\theta}_n^{L_2})\Big)\Big|_{t=0}\nonumber\\
&=&\frac{2}{n}\sum_{i=1}^n\{\hat{\zeta}_n(x_i)-y^p_i\}\frac{\partial y^s}{\partial \theta_j}(x_i,\hat{\theta}_n^{L_2})+2\lambda_n\Big\langle\hat{\zeta}_n,\frac{\partial y^s}{\partial \theta_j}(\cdot,\hat{\theta}_n^{L_2})\Big\rangle_{\mathcal{N}_\Phi(\Omega)}\nonumber\\
&=&\frac{2}{n}\sum_{i=1}^n\{\hat{\zeta}_n(x_i)-\zeta(x_i)\}\frac{\partial y^s}{\partial \theta_j}(x_i,\hat{\theta}_n^{L_2})-\frac{2}{n}\sum_{i=1}^n e_i\frac{\partial y^s}{\partial \theta_j}(x_i,\hat{\theta}_n^{L_2})\nonumber\\
&&+2\lambda_n\Big\langle\hat{\zeta}_n,\frac{\partial y^s}{\partial \theta_j}(\cdot,\hat{\theta}_n^{L_2})\Big\rangle_{\mathcal{N}_\Phi(\Omega)}=:2(C_n+D_n+E_n).\label{zero}
\end{eqnarray}

First, we consider $C_n$.
Let $A_i(g,\theta)=\{g(x_i)-\zeta(x_i)\}\frac{\partial y^s}{\partial \theta_j}(x_i,\theta)$ for $(g,\theta)\in\mathcal{N}_\Phi(\Omega,\rho)\times U$ with some $\rho>0$ to be specified later. Then $E[A_i(g,\theta)]=\int_{\Omega}\{g(z)-\zeta(z)\}\frac{\partial y^s}{\partial \theta_j}(z,\theta) d z$.
Define the empirical process
\begin{eqnarray*}
E_{1 n}(g,\theta)=n^{-1/2}\sum_{i=1}^n\{A_i(g,\theta)-E[A_i(g,\theta)]\}.
\end{eqnarray*}
By B1, $\mathcal{N}_\Phi(\Omega,k)$ is Donsker. Thus $\mathcal{F}_1=\{g-\zeta:g\in\mathcal{N}_\Phi(\Omega,\rho)\}$ is also Donsker. Condition A5 implies that $\mathcal{F}_2=\{\frac{\partial y^s}{\partial \theta_j}(\cdot,\theta):\theta\in U\}$ is Donsker. Since both $\mathcal{F}_1$ and $\mathcal{F}_2$ are uniformly bounded, the product class $\mathcal{F}_1\times\mathcal{F}_2$ is also Donsker. 
For theorems on Donsker classes, we refer to \cite{kosorok2008introduction} and the references therein.
Thus the asymptotic equicontinuity property holds, which suggests that (see Theorem 2.4 of \cite{mammen1997penalized}) for any $\xi>0$ there exists a $\delta>0$ such that
\begin{eqnarray*}
\limsup_{n\rightarrow\infty}P\left(\sup_{f\in\mathcal{F}_1\times\mathcal{F}_2,\|f\|\leq \delta}\left|\frac{1}{\sqrt{n}}\sum_{i=1}^n(f(x_i)-E(f(x_i)))\right|> \xi\right)<\xi,
\end{eqnarray*}
where $\|\cdot\|$ is defined as $\|f\|^2:=E[f(x_i)]^2$. This implies that for all $\xi>0$ there exists a $\delta>0$ such that
\begin{eqnarray}
\limsup_{n\rightarrow\infty}P\left(\sup_{g\in\mathcal{N}_\Phi(\Omega,\rho),\theta\in U,\|g-\zeta\|_{L_2(\Omega)}\leq \delta} |E_{1 n}(g,\theta)|>\xi\right)<\xi.\label{equicontinuous}
\end{eqnarray}
Now fix $\varepsilon>0$. Condition B3 implies that there exists $\rho_0>0$, such that $P(\|\hat{\zeta}_n\|_{\mathcal{N}_\Phi(\Omega)}>\rho_0)\leq \varepsilon/3$. Choose $\delta_0$ to be a possible value of $\delta$ satisfying (\ref{equicontinuous}) with $\rho=\rho_0$ and $\xi=\varepsilon/3$. Define
\begin{eqnarray*}
\hat{\zeta}_n^\circ:=
\begin{cases}
\hat{\zeta}_n & \text{ if } \|\hat{\zeta}_n\|_{\mathcal{N}_\Phi(\Omega)}\leq\rho_0 \text{ and } \|\hat{\zeta}_n-\zeta\|_{L_2(\Omega)}\leq\delta_0\\
\zeta & \text{ elsewise }
\end{cases}.
\end{eqnarray*}
Therefore, for sufficiently large $n$ we have
\begin{eqnarray*}
&&P(|E_{1 n}(\hat{\zeta}_n,\hat{\theta}_n^{L_2})|>\varepsilon)\\
&\leq& P(|E_{1 n}(\hat{\zeta}_n^\circ,\hat{\theta}_n^{L_2})|>\varepsilon)+P(\|\hat{\zeta}_n\|_{\mathcal{N}_\Phi(\Omega)}>\rho_0) +P(\|\hat{\zeta}_n-\zeta\|_{L_2(\Omega)}>\delta_0)\\
&\leq& P(|E_{1 n}(\hat{\zeta}_n^\circ,\hat{\theta}_n^{L_2})|>\varepsilon/3)+\varepsilon/3+\varepsilon/3\\
&\leq& P\left(\sup_{g\in\mathcal{N}(\Omega,\rho_0),\theta\in U, \|g-\zeta\|_{L_2(\Omega)}\leq \delta_0}|E_{1 n}(g,\theta)|>\varepsilon/3\right)+\varepsilon/3+\varepsilon/3\\
&\leq& \varepsilon,
\end{eqnarray*}
where the first and the third inequalities follows from the definition of $\hat{\zeta}_n^\circ$; the second inequality follows from B2; the last inequality follows from (\ref{equicontinuous}). This implies that $E_{1 n}(\hat{\zeta}_n,\hat{\theta}_n^{L_2})$ tends to zero in probability.
Thus we have
\begin{eqnarray*}
o_p(1)&=&E_{1 n}(\hat{\zeta}_n,\hat{\theta}_n^{L_2})=n^{-1/2}\sum_{i=1}^n\{\hat{\zeta}_n(x_i)-\zeta(x_i)\}\frac{\partial y^s}{\partial \theta_j}(x_i,\hat{\theta}_n^{L_2})\\&&-n^{1/2}\int_\Omega\{\hat{\zeta}(z)-\zeta(z)\}\frac{\partial y^s}{\partial \theta_j}(z,\hat{\theta}_n^{L_2})d z\\
&=&n^{1/2}C_n-n^{1/2}\int_\Omega\{\hat{\zeta}(z)-\zeta(z)\}\frac{\partial y^s}{\partial \theta_j}(z,\hat{\theta}_n^{L_2})d z,
\end{eqnarray*}
which implies
\begin{eqnarray}
C_n=\int_{\Omega}\{\hat{\zeta}_n(z)-\zeta(z)\}\frac{\partial y^s}{\partial \theta_j}(z,\hat{\theta}_n^{L_2}) d z+o_p(n^{-1/2}).\label{Cn0}
\end{eqnarray}
By substituting (\ref{eq theta}) to (\ref{Cn0}) and using A2, we can apply the Taylor expansion to (\ref{Cn0}) at $\theta^*$ and obtain
\begin{eqnarray}
C_n&=&\int_{\Omega}\{y^s(z,\hat{\theta}_n^{L_2})-\zeta(z)\}\frac{\partial y^s}{\partial \theta_j}(z,\hat{\theta}_n^{L_2}) d z+o_p(n^{-1/2})\nonumber\\
&=&\left\{\frac{1}{2}\int_{\Omega}\frac{\partial^2}{\partial\theta^{T}\partial \theta_j}\left(y^s(z,\tilde{\theta}_n)-\zeta(z)\right)^2 d z\right\} (\hat{\theta}_n^{L_2}-\theta^*)\nonumber\\&&+o_p(n^{-1/2}),\label{Cn}
\end{eqnarray}
where $\tilde{\theta}_n$ lies between $\hat{\theta}_n$ and $\theta^*$. By the consistency of $\hat{\theta}_n^{L_2}$, we have $\tilde{\theta}_n\operatorname*{\rightarrow}\limits^p\theta^*$. This implies that
\begin{eqnarray}
&&\int_{\Omega}\frac{\partial^2}{\partial\theta^{T}\partial \theta}\left(y^s(z,\tilde{\theta}_n)-\zeta(z)\right)^2 d z\nonumber\\&\operatorname*{\rightarrow}\limits^p& \int_{\Omega}\frac{\partial^2}{\partial\theta^{T}\partial \theta}\left(y^s(z,\theta^*)-\zeta(z)\right)^2 d z=V.\label{Vn}
\end{eqnarray}

Now we consider $D_n$. Define the empirical process
\begin{eqnarray*}
&&E_{2 n}(\theta)\\&=&n^{-1/2}\sum_{i=1}^n\left\{e_i\frac{\partial y^s}{\partial \theta_j}(x_i,\theta)-e_i\frac{\partial y^s}{\partial \theta_j}(x_i,\theta^*)-E\left[e_i\frac{\partial y^s}{\partial \theta_j}(x_i,\theta)-e_i\frac{\partial y^s}{\partial \theta_j}(x_i,\theta^*)\right]\right\}\\&=&n^{-1/2}\sum_{i=1}^n\left\{ e_i\frac{\partial y^s}{\partial \theta_j}(x_i,\theta)- e_i\frac{\partial y^s}{\partial \theta_j}(x_i,\theta^*)\right\},
\end{eqnarray*}
where $\theta\in U$. By A5, the set $\{f_\theta\in C(\mathbf{R}\times\Omega):f_\theta(e,x)=e\frac{\partial y^s}{\partial \theta_j}(x,\theta)- e\frac{\partial y^s}{\partial \theta_j}(x,\theta^*),\theta\in U\}$ is a Donsker class. This ensures that $E_{2 n}(\cdot)$ weakly converges in $L_\infty(U)$ to a tight Guassian process, denoted by $G(\cdot)$. Without loss of generality, we assume that $G(\cdot)$ has continuous sample paths. We note that $G(\theta^*)=0$ because $E_{2 n}(\theta^*)$ for all $n$. Then as a consequence of the consistency of $\hat{\theta}^{L_2}_n$ and the continuous mapping theorem \citep{van2000asymptotic}, $E_{2 n}(\hat{\theta}_n^{L_2})\operatorname*{\rightarrow}\limits^p G(\theta^*)=0$, which gives
\begin{eqnarray}
D_n=\frac{1}{n}\sum_{i=1}^n e_i\frac{\partial y^s}{\partial \theta_j}(x_i,\theta^*)+o_p(n^{-1/2}).\label{Dn}
\end{eqnarray}

Finally we estimate $C4$. By A5, B2 and B3,
By C2, C3 and C4,
\begin{eqnarray}
E_n\leq \lambda_n\|\hat{\zeta}\|_{\mathcal{N}_\Phi(\Omega)}\Big\|\frac{\partial y^s}{\partial \theta_j}(\cdot,\hat{\theta})\Big\|_{\mathcal{N}_\Phi(\Omega)}=o_p(n^{-1/2}).\label{En}
\end{eqnarray}

By combining (\ref{zero}), (\ref{Cn}), (\ref{Vn}), (\ref{Dn}) and (\ref{En}), we prove the desired result.
\end{proof}

Theorem \ref{Th AN} implies the asymptotic normality of $\sqrt{n}(\hat{\theta}^{L_2}_n-\theta^*)$, provided that
\begin{eqnarray}
W:=E\left[\frac{\partial y^s}{\partial \theta}(x_i,\theta^*)\frac{\partial y^s}{\partial \theta^\text{T}}(x_i,\theta^*)\right]\label{W}
\end{eqnarray}
is positive definite. Specifically,
\begin{eqnarray}
\sqrt{n}(\hat{\theta}^{L_2}_n-\theta^*)\operatorname*{\rightarrow}^d N(0,4\sigma^2 V^{-1}WV^{-1}).\label{L2asymvariance}
\end{eqnarray}

\subsection{Semiparametric Efficiency}
In this section, we discuss the efficiency of the proposed $L_2$ calibration. It will be shown that, as a semiparametric method, the $L_2$ calibration method reaches the highest possible efficiency if the measurement errors follow a normal distribution.

In statistics, a parametric model is one whose parameter space is finite dimensional, while a nonparametric model is one with an infinite dimensional parameter space. The definition of semiparametric models is, nevertheless, more complicated. Refer to \cite{groeneboom1992information,bickel1993efficient} for details. In simple terms, a semiparametric problem has an infinite dimensional parameter space but the parameter of interest in this problem is only finite dimensional. The calibration problems under consideration are semiparametric. To see this, consider the calibration model given by (\ref{lpls}) and (\ref{l2projection}). The parameter space of model (\ref{lpls}) contains an infinite dimensional function space which covers $\zeta$. On the other hand, the parameter of interest is $\theta^*$ in (\ref{l2projection}), which is $q$-dimensional.

Now we briefly review the estimation efficiency in semiparametric problems. For details, we refer to \cite{bickel1993efficient,kosorok2008introduction}. Let $\Xi$ be an infinite dimensional parameter space whose true value is denoted by $\xi_0$. Denote the feature of interest as $\nu(\xi_0)$ with a known map $\nu:\Xi\mapsto\mathbf{R}^d$.
Suppose $T_n$ is an estimator for $\nu(\xi_0)$ based on $n$ independent samples and that $\sqrt{n}(T_n-\nu(\xi_0))$ is asymptotically normal.
Now let $\Xi_0$ be an arbitrary finite dimensional subset of $\Xi$ satisfying $\xi_0\in\Xi_0$. We consider the statistical estimation problem with the same observed data but with the parameter space $\Xi_0$. Under this parametric assumption and some other regularity conditions, an efficient estimator can be obtained by using the maximum likelihood (ML) method, denoted by $S_n^{\Xi_0}$. Since the construction of $S_n^{\Xi_0}$ uses more assumptions than $T_n$, the asymptotic variance of $S_n^{\Xi_0}$ should be less than or equal to that of $T_n$. We call $T_n$ \textit{semiparametric efficient} if there exists a $\Xi_0$ such that $S_n^{\Xi_0}$ has the same asymptotic variance as $T_n$.



For the calibration problem given by (\ref{lpls}) and (\ref{l2projection}),
consider the following $q$-dimensional parametric model indexed by $\gamma$:
\begin{eqnarray}
\zeta_\gamma(\cdot)=\zeta(\cdot)+\gamma^\text{T}\frac{\partial y^s}{\partial \theta}(\cdot,\theta^*),\label{gamma}
\end{eqnarray}
with $\gamma\in\mathbf{R}^q$. Then (\ref{lpls}) and (\ref{gamma}) form a linear regression model. Regarding (\ref{lpls}), the true value of $\gamma$ is $\gamma_0=0$. Suppose that $e_i$ in (\ref{lpls}) follows $N(0,\sigma^2)$ with an unknown $\sigma^2$. Under the regularity conditions of Theorem \ref{Th AN}, the ML estimator for observations $\{(x_i,y_i)\}_{i=1}^n$ is the least squares estimator, with the asymptotic expression
\begin{eqnarray}
\hat{\gamma}_n=\frac{1}{n}W^{-1}\sum_{i=1}^n e_i\frac{\partial y^s}{\partial \theta}(x_i,\theta^*)+o_p(n^{-1/2}),\label{MLE}
\end{eqnarray}
where $W$ is defined in (\ref{W}). Then a natural estimator for $\theta^*$ in (\ref{l2projection}) is
\begin{eqnarray}
\hat{\theta}_n=\operatorname*{argmin}_{\theta\in\Theta}\|\zeta_{\hat{\gamma}_n}(\cdot)-y^s(\cdot,\theta)\|_{L_2(\Omega)}.\label{parametric}
\end{eqnarray}
Again, we simplify the problem in (\ref{parametric}) by assuming that $y^s$ is a known function. The asymptotic variance of $\hat{\theta}_n$ can be obtained by the delta method. As in A1, assume that $x_i$ follows the uniform distribution over $\Omega$. Then $\|f\|_{L_2(\Omega)}^2=E f^2(x_i)$ for all $f$. Define
\begin{eqnarray}
\theta(t)=\operatorname*{argmin}_{\theta\in\Theta}E[\zeta_{t}(x_i)-y^s(x_i,\theta)]^2,\label{thetat}
\end{eqnarray}
for each $t$ near 0. Let
\begin{eqnarray*}
&&\Psi(\theta,t)=\frac{\partial}{\partial\theta}E[\zeta_{t}(x_i)-y^s(x_i,\theta)]^2\\
&=& \frac{\partial}{\partial\theta}E\left[\zeta(x_i)-t^\text{T}\frac{\partial y^s}{\partial \theta}(x_i,\theta^*)-y^s(x_i,\theta)\right]^2.
\end{eqnarray*}
Then (\ref{thetat}) implies $\Psi(\theta(t),t)=0$ for all $t$ near 0. From the implicit function theorem, we have
\begin{eqnarray}
&&\frac{\partial \theta(t)}{\partial t^\text{T}}\Big|_{t=0}=-\Big(\frac{\partial \Psi}{\partial \theta^\text{T}}(\theta^*,0)\Big)^{-1}\frac{\partial \Psi}{\partial t^\text{T}}(\theta^*,0)\nonumber\\
&=&-\Big(E\frac{\partial^2}{\partial \theta\partial\theta^\text{T}}[\zeta(x_i)-y^s(x_i,\theta^*)]^2\Big)^{-1}2 E\left[\frac{\partial y^s}{\partial \theta^\text{T}}(x_i,\theta^*)\frac{\partial y^s}{\partial \theta}(x_i,\theta^*)\right]\nonumber\\&=&-2V^{-1}W.\label{partial}
\end{eqnarray}
By the delta method,
\begin{eqnarray}
\hat{\theta}_n-\theta^*=\theta(\hat{\gamma}_n)-\theta(0)=\frac{\partial \theta(t)}{\partial t^\text{T}}\Big|_{t=0}\hat{\gamma}_n+o_p(n^{-1/2}),
\end{eqnarray}
which, together with (\ref{MLE}) and (\ref{partial}), yields
\begin{eqnarray}
\hat{\theta}_n-\theta^*=-2 V^{-1}\sum_{i=1}^n e_i\frac{\partial y^s}{\partial \theta}(x_i,\theta^*)+o_p(n^{-1/2}).\label{hard}
\end{eqnarray}
Noting that the asymptotic expression of the $L_2$ calibration given by (\ref{normality}) has the same form as the ML estimator for the parametric model in (\ref{hard}), we obtain the following theorem.

\begin{theorem}\label{Th efficiency}
Under the assumptions of Theorem \ref{Th AN}, if $e_i$ in (\ref{lpls}) follows a normal distribution, then the $L_2$ calibration (\ref{L2stochastic}) is semiparametric efficient.
\end{theorem}

Since the normal distribution is commonly used to model the random error in physical experiments (see e.g. \citealp{wu2011experiments}) and the calibration for computer experiments (see e.g. \citealp{kennedy2001bayesian}). Theorem \ref{Th efficiency} suggests that the proposed method is efficient for many practical problems.
For non-normal error distributions, the ML estimator does not agree with the least squares estimator. Thus the ML estimator cannot be expressed by (\ref{MLE}). Consequently, the $L_2$ calibration defined by (\ref{smoothingn}) and (\ref{L2stochastic}) is not semiparametric efficient. However, if the random error is from a parametric model, the proposed $L_2$ calibration can be modified to achieve the semiparametric efficiency. Denote the likelihood function of $e_i$ by $l(\beta;e_i)$ with $\beta\in \mathcal{B}$, i.e, $e_i$ has a density $l(\beta_0;\cdot)$ for some unknown $\beta_0\in\mathcal{B}$. Suppose $L(\cdot;x):=\log l(\cdot;x)$ is convex for all $x$. Then the penalized ML estimator for $\zeta$ is
\begin{eqnarray}
\hat{\zeta}^{ML}:=\operatorname*{argmin}_{\beta\in\mathcal{B},f\in\mathcal{N}_\Phi(\Omega)}\frac{1}{n}\sum_{i=1}^n L(\beta;y_i-f(x_i))+ \lambda\|f\|^2_{\mathcal{N}_\Phi(\Omega)}.
\end{eqnarray}
Then define the modified $L_2$ calibration as
\begin{eqnarray}
\hat{\theta}^{ML}:=\operatorname*{argmin}_{\theta\in\Theta}\|\hat{\zeta}^{ML}(\cdot)-\hat{y}^s(\cdot,\theta)\|_{L_2(\Omega)}.
\end{eqnarray}
By using similar arguments, it can be proved that, under some regularity conditions, $\hat{\theta}^{ML}$ is semiparametric efficient. For a related discussion, we refer to \cite{shen1997methods}.

\section{Ordinary Least Squares}
In this section, we will study an alternative method, namely, the ordinary least squares (OLS) calibration. There are several versions of the OLS method discussed in statistics and applied mathematics for calibration problems and inverse problems \citep[e.g.,][]{joseph2009statistical,evans2002inverse}. Here we consider a general form, which is apparently new but covers the existing versions. As before, let $\hat{y}^s_n$ be a sequence of surrogate models for $y^s$.
Define the OLS estimator for the calibration parameter as
\begin{eqnarray}
\hat{\theta}^{OLS}_n=\operatorname*{argmin}_{\theta\in\Theta}\sum_{i=1}^n\left(y^p_i-\hat{y}^s_n(x_i,\theta)\right)^2,\label{OLS}
\end{eqnarray}
where $x_i$'s and $y_i$'s are from the model (\ref{lpls}).

Obviously, the OLS calibration is a natural choice when there is no difference between the true process and the optimal computer output, i.e., $\zeta(\cdot)=y^s(\cdot,\theta^*)$. However, we are particularly interested in the asymptotic behavior of the OLS calibration when $\zeta(\cdot)$ and $y^s(\cdot,\theta^*)$ are different.

Analogous to Theorem \ref{Th AN} for the $L_2$ calibration, we have the following theorem on the asymptotic behavior of the OLS calibration.


\begin{theorem}\label{Th OLS}
In addition to conditions A1-A4 and C1-C2, suppose that there exists a neighborhood $U$ of $\theta^*$, such that $y^s(x,\cdot)\in C^{2,1}(U)$ for all $x\in\Omega$, where $C^{2,1}$ denotes the space of functions whose second derivatives are Lipschitz. Then
\begin{eqnarray*}
\hat{\theta}^{OLS}_n-\theta^*=V^{-1}\left\{\frac{1}{n}\sum_{i=1}^n\frac{\partial}{\partial\theta}(y^p_i-y^s(x_i,\theta^*))^2\right\}+o_p(n^{-1/2}).
\end{eqnarray*}
\end{theorem}

\begin{proof}
First we prove $\hat{\theta}^{OLS}_n\operatorname*{\rightarrow}\limits^p \theta^*$. By condition A2, it suffices to show that
\begin{eqnarray}
\sup_{\theta\in\Theta}\left|\frac{1}{n}\sum_{i=1}^n\left(y^p_i-\hat{y}^s_n(x_i,\theta)\right)^2- (\|\zeta(\cdot)-y^s(\cdot,\theta)\|^2_{L_2(\Omega)}+\sigma^2)\right|\operatorname*{\rightarrow}\limits^p 0.\label{LLN}
\end{eqnarray}
Note that
\begin{eqnarray}
&&\left|\frac{1}{n}\sum_{i=1}^n\left(y^p_i-\hat{y}^s_n(x_i,\theta)\right)^2-\frac{1}{n}\sum_{i=1}^n\left(y^p_i-y^s(x_i,\theta)\right)^2\right| \nonumber\\
&=&\left|\frac{1}{n}\sum_{i=1}^n(y^s(x_i,\theta)-\hat{y}^s(x_i,\theta))(2 y^p_i-y^s(x_i,\theta)-\hat{y}^s(x_i,\theta))\right|\nonumber\\
&\leq&\|y^s-\hat{y}^s\|_{L_\infty(\Omega)}\left(\frac{1}{n}\sum_{i=1}^n 2 (y^p_i-y^s(x_i,\theta))+\|y^s-\hat{y}^s\|_{L_\infty(\Omega)}\right) \label{distanceemulator}
\end{eqnarray}
Since $\Theta$ is compact, the uniform law of large numbers \citep{van1996weak} implies
\begin{eqnarray}
\sup_{\theta\in\Theta}\left|\frac{1}{n}\sum_{i=1}^n (y^p_i-y^s(x_i,\theta))-E[y^p_i-y^s(x_i,\theta)]\right|\operatorname*{\rightarrow}\limits^p 0,\label{uniformconv1}
\end{eqnarray}
and
\begin{eqnarray}
\sup_{\theta\in\Theta}\left|\frac{1}{n}\sum_{i=1}^n (y^p_i-y^s(x_i,\theta))^2-E[y^p_i-y^s(x_i,\theta)]^2\right|\operatorname*{\rightarrow}\limits^p 0.\label{uniformconv2}
\end{eqnarray}
Direct calculations give
\begin{eqnarray}
E[y^p_i-y^s(x_i,\theta)]^2=\int_{\Omega}(\zeta(z)-y^s(z,\theta))^2 d z +\sigma^2,\label{moment}
\end{eqnarray}
which, together with (\ref{distanceemulator}), (\ref{uniformconv1}), (\ref{uniformconv2}) and (\ref{moment}), proves (\ref{LLN}).

By the definition (\ref{OLS}), condition A2 and the consistency of $\hat{\theta}_n^{OLS}$, we have
\begin{eqnarray*}
0&=&\frac{\partial}{\partial \theta}\left\{\frac{1}{n}\sum_{i=1}^n\left(y^p_i-\hat{y}^s_n(x_i,\hat{\theta}_n^{OLS})\right)^2\right\}\\
&=&\frac{2}{n}\sum_{i=1}^n\frac{\partial\hat{y}^s_n}{\partial \theta}(x_i,\hat{\theta}_n^{OLS})\left\{\hat{y}^s_n(x_i,\hat{\theta}_n^{OLS})-y_i^p\right\},
\end{eqnarray*}
which, together with the law of large numbers and conditions C1 and C2, yields
\begin{eqnarray*}
o_p(n^{-1/2})
&=&\frac{2}{n}\sum_{i=1}^n\frac{\partial y^s}{\partial \theta}(x_i,\hat{\theta}_n^{OLS})\left\{y^s(x_i,\hat{\theta}_n^{OLS})-y_i^p\right\}\nonumber\\
&=&\frac{1}{n}\sum_{i=1}^n\frac{\partial}{\partial \theta}\left(y^p_i-y^s(x_i,\hat{\theta}_n^{OLS})\right)^2.
\end{eqnarray*}
The remainder of the proof follows from some direct calculations using the standard asymptotic theory for Z-estimators \citep{van2000asymptotic}.
\end{proof}

Theorem \ref{Th OLS} shows that the OLS calibration is consistent even if the computer code is imperfect. Compared with the $L_2$ calibration, the OLS calibration is computationally more efficient. Also, the OLS calibration does not require tuning, while in the $L_2$ calibration the value of the tuning parameter $\lambda$ in (\ref{smoothing}) needs to be determined. However, according to Theorem \ref{Th OLS}, the asymptotic variance of the OLS calibration does not reach the semiparametric lower bound given by (\ref{hard}).

We now study the conditions under which the $L_2$ calibration and the OLS calibration are asymptotically equivalent.
Let $\Sigma_1=4\sigma^2 W$. Then the asymptotic variance of the $L_2$ calibration given by (\ref{L2asymvariance}) is $V^{-1}\Sigma_1 V^{-1}$. Let
\begin{eqnarray}
\Sigma_2&=&E\left[\frac{\partial}{\partial\theta}(y^p_i-y^s(x_i,\theta^*))^2\right]^2\nonumber\\
&=&4E\left[(e_i+\zeta(x_i)-y^s(x_i,\theta^*))^2 \frac{\partial y^s}{\partial\theta}(x_i,\theta^*)\frac{\partial y^s}{\partial\theta^\text{T}}(x_i,\theta^*)\right]\nonumber\\
&=&4\sigma^2 W+4E\left[(\zeta(x_i)-y^s(x_i,\theta^*))^2\frac{\partial y^s}{\partial\theta}(x_i,\theta^*)\frac{\partial y^s}{\partial\theta^\text{T}}(x_i,\theta^*)\right].\label{OLSasymvariance}
\end{eqnarray}
Then Theorem \ref{Th OLS} shows that the asymptotic variance for the OLS calibration is $V^{-1}\Sigma_2 V^{-1}$. From (\ref{OLSasymvariance}), it is seen that $\Sigma_2-\Sigma_1\geq 0$. Additionally, $\Sigma_1=\Sigma_2$ if and only if
\begin{eqnarray}
E\left[(\zeta(x_i)-y^s(x_i,\theta^*))^2\frac{\partial y^s}{\partial\theta}(x_i,\theta^*)\frac{\partial y^s}{\partial\theta^\text{T}}(x_i,\theta^*)\right]=0.\label{equcondition}
\end{eqnarray}

Suppose $\frac{\partial y^s}{\partial\theta}(x,\theta^*)\neq 0$ for all $x\in\Omega$. Then (\ref{equcondition}) holds only if $\zeta(x)=y^s(x,\theta^*)$ for almost every $x\in\Omega$, i.e., there exists a perfect computer model. In this case, the OLS calibration has the same asymptotic distribution as the $L_2$ calibration.
However, as suggested by \cite{kennedy2001bayesian}, the bias between $y^s(\cdot,\theta^*)$ and $\zeta(\cdot)$ can be large in practical situations. Thus in general the OLS calibration is less efficient than the $L_2$ calibration.

\section{Numerical Studies}

In this section, we compare the numerical behaviors of three methods for the estimation of the calibration parameters: the $L_2$ calibration, the OLS calibration, and a version of the method proposed by \cite{kennedy2001bayesian}. The original version of the Kennedy-O'Hagan (abbreviated as KO) method is a Bayesian approach. In order to compare with the proposed frequentist methods, we consider the frequentist version of the KO method stated in \cite{tuo2014calibration}, where the maximum likelihood estimation is used.

\subsection{Example 1: perfect computer model}

Suppose the true process is
\begin{eqnarray}
\zeta(x)=\exp(x/10)\sin x,\label{zetanumerical}
\end{eqnarray}
for $x\in\Omega=(0,2\pi)$. The physical observations are given by
\begin{eqnarray}
y^p_i=\zeta(x_i)+e_i,\label{phisicalnumberical}
\end{eqnarray}
with
\begin{eqnarray}
x_i=2\pi i/50, e_i\sim N(0,\sigma^2), \text{ for } i=0,\ldots,50.\label{samplenumerical}
\end{eqnarray}
We will consider two levels of $\sigma^2$ (with $\sigma^2=0.1$ and $\sigma^2=1$) so that the numerical stability of the methods with different noise levels is investigated.

Suppose the computer output is
\begin{eqnarray}
y^s(x,\theta)=\zeta(x)-|\theta+1|(\sin \theta x+\cos\theta x).\label{emulatornumerical}
\end{eqnarray}
Then we have $\zeta(\cdot)=y^s(\cdot,-1)$. Thus $\theta^*=-1$. And there is no discrepancy between $\zeta(\cdot)$ and $y^s(\cdot,\theta^*)$, i.e., the computer model is perfect. For simplicity, we suppose that (\ref{emulatornumerical}) is a known function so that we do not need an emulator for it.

We conducted 1000 random simulations to examine the performance of the $L_2$ calibration, the OLS calibration, and the KO calibration for $\sigma^2=0.1$ and $\sigma^2=1$ respectively. For the $L_2$ calibration and the KO calibration, the Gaussian correlation family $\Phi(x_1,x_2)=\exp\{-\phi(x_1-x_2)^2\}$ is used with the model parameter $\phi$ chosen by the cross-validation method \citep{santner2003design,Rasmussen2006gaussian}. The tuning parameters in the nonparametric regression is selected by the generalized cross validation \citep{wahba1990spline}.

Table \ref{Tab eg1} shows the simulation results. The results for $\sigma^2=0.1$ and $\sigma^2=1$ are given in columns 2-3 and 4-5 respectively. The true values of $\theta^*$ are given in the second row. The last three rows give the mean value and the mean square error (MSE) over 1000 random simulations for the three methods.

\begin{table*}[h]
\caption{Numerical comparison for perfect computer model. MSE = mean square error.}\label{Tab eg1}
\centering
\begin{tabular}{|c|c|c|c|c|}
  \hline
   & \multicolumn{2}{|c|}{$\sigma^2=0.1$} & \multicolumn{2}{|c|}{$\sigma^2=1$} \\
  \hline True Value & \multicolumn{2}{|c|}{-1}  & \multicolumn{2}{|c|}{-1}   \\
  \hline  & Mean & MSE & Mean & MSE \\
  \hline $L_2$ & -0.9990 & $6.497\times 10^{-5}$ & -0.8876 & 0.0906 \\
  \hline OLS & -0.9999 & $1.160\times 10^{-4}$ & -0.9306 & 0.0908 \\
  \hline KO & -0.9993 & $8.065\times 10^{-5}$ & -0.9325 & 0.0468 \\
  \hline
\end{tabular}
\end{table*}

It can be seen from Table \ref{Tab eg1} that all three methods give good estimation results in this example. The good performance of the KO method is not surprising because the computer model here is perfect. In their theoretical study on the KO method with deterministic physical experiments, \cite{tuo2014calibration} obtained the limiting value of the KO method under certain conditions. Using Theorem 1 of \cite{tuo2014calibration}, it can be seen that, for deterministic physical experiments, the Kennedy-O'Hagan method would be consistent if the computer model is perfect. The simulation results in this example suggest that this statement may also hold for stochastic physical systems.

\subsection{Example 2: imperfect computer model}

Now we consider an example with an imperfect computer model. Suppose the true process and the physical observations are the same as in Example 1, given by (\ref{zetanumerical}), (\ref{phisicalnumberical}), and (\ref{samplenumerical}). Suppose the computer model is
\begin{eqnarray}
y^s(x,\theta)=\zeta(x)-\sqrt{\theta^2-\theta+1}(\sin \theta x+\cos\theta x).\label{ysimperfect}
\end{eqnarray}
As in Example 1, we suppose $y^s$ is known. From (\ref{ysimperfect}), it can be seen that there does not exist a real number $\theta$ satisfying $y^s(\cdot,\theta)=\zeta(\cdot)$, because the quadratic function $\theta^2-\theta+1$ is always positive. Thus, this computer model is imperfect.

The $L_2$ discrepancy between the computer model and the physical model has an explicit form:
\begin{eqnarray}
\|\zeta-y^s(\cdot,\theta)\|^2_{L_2(\Omega)}=(\theta^2-\theta+1)\left(2\pi-\frac{\cos(4\pi\theta)-1}{2\theta}\right),\label{discrepancy}
\end{eqnarray}
with a continuous extension at $\theta=0$. Figure \ref{Fig:discrepancy} plots the function (\ref{discrepancy}) with $-2<\theta<2$. Numerical optimization shows that the minimizer of (\ref{discrepancy}) is $\theta^*\approx -0.1789$.

\begin{figure}
  \centering
  \includegraphics[width=0.5\textwidth]{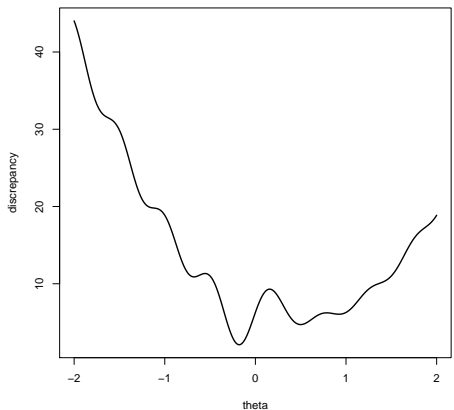}\\
  \caption{$L_2$ discrepancy function in Example 2.}\label{Fig:discrepancy}
\end{figure}

As in Example 1, we conducted 1000 random simulations to compare the $L_2$ calibration, the OLS calibration, and the KO calibration. We keep the remaining setup of this experiment the same as in Example 1. The mean value and standard deviation (SD) over 1000 simulations are shown in Table \ref{Tab eg2}.

\begin{table*}[h]
\caption{Numerical comparison for imperfect computer model. SD = standard deviation.}\label{Tab eg2}
\centering
\begin{tabular}{|c|c|c|c|c|}
  \hline
   & \multicolumn{2}{|c|}{$\sigma^2=0.1$} & \multicolumn{2}{|c|}{$\sigma^2=1$} \\
  \hline True Value & \multicolumn{2}{|c|}{-0.1789}  & \multicolumn{2}{|c|}{-0.1789}   \\
  \hline  & Mean & SD & Mean & SD \\
  \hline $L_2$ & -0.1792 & $2.665\times 10^{-3}$ & -0.1773 & 0.0711 \\
  \hline OLS & -0.1770 & $2.674\times 10^{-3}$ & -0.1684 & 0.1060 \\
  \hline KO & -0.1224 & $7.162\times 10^{-3}$ & 0.0034 & 0.3244 \\
  \hline
\end{tabular}
\end{table*}

It can be seen from Table \ref{Tab eg2} that the $L_2$ calibration and the OLS calibration
outperform the KO calibration. Furthermore, the mean value of
the KO estimator changes a lot as $\sigma^2$ changes. This is undesirable because
a good estimator should not be sensitive to random error for large samples.
Table \ref{Tab eg2} also shows that the standard deviation of the $L_2$ calibration is
smaller than that of the OLS calibration. This agrees with our theoretical
analysis, which shows that the $L_2$ calibration is more efficient than the OLS
calibration for imperfect computer models. Overall, the KO calibration underperforms the $L_2$  calibration or the OLS calibration.

\section{Concluding Remarks and Further Discussions}

In this work, we extend the framework established in \cite{tuo2014calibration} to stochastic physical systems. We propose a novel method, called the $L_2$ calibration, and prove its asymptotic normality and semiparametric efficiency. We also study the OLS method and prove that it is consistent but not efficient.
Although the OLS calibration is computationally less costly, the $L_2$ calibration should be seriously considered because of its high estimation efficiency.
By using a more efficient estimator, fewer physical trials are needed to achieve the same estimation efficiency.
In most practical problems, physical experiments are more expensive to run. Therefore it would be worthwhile to save the physical runs by doing more computation. Thus we recommend using the $L_2$ calibration over the OLS calibration.

Because of the identifiability problem in calibration, we define the purpose of calibration as that of finding the $L_2$ projection, i.e., the parameter value which minimizes the discrepancy between the true process and the computer output under the $L_2$ norm. Noting that the ``true'' value of the calibration parameter in our framework depends on the choice of the norm, one may also consider the asymptotic results for calibration under a different norm. After some calculations, it can be shown that the main results of this work still hold if the new norm is equivalent to the $L_2$ norm. However, if a norm that is not equivalent to the $L_2$ norm, such as the $L_\infty$ norm, is used, the idea in the proof of Theorem 1 will not work. We believe that, for those norms, there do not exist estimators with convergence rate $O(n^{-1/2})$. This will require further work.

We have reported the asymptotic properties of the $L_2$ calibration under the random design, i.e., $x_i$ are sampled independently from the uniform distribution. Given the fact that many physical experiments are conducted under fixed designs (see books by \citealp{box2005statistics} and \citealp{wu2011experiments}), the results for calibration for fixed designs need further investigation.

\appendix

\section*{Acknowledgements}
The authors are grateful to the Associate Editor and the referees for helpful comments.

\bibliographystyle{imsart-nameyear}
\bibliography{calibration}

\end{document}